\documentclass[onecolumn,noversion,nonote]{cdcarticle}

\def\MODE{3}

\pdfoutput=1

\usepackage{amsmath,amsthm,amssymb}
\usepackage{tikz}
\usepackage{graphicx}
\usepackage[hidelinks]{hyperref}

\newtheorem{thm}{Theorem}
\newtheorem{problem}[thm]{Problem}
\newtheorem{lem}[thm]{Lemma}

\newtheorem{cor}[thm]{Corollary}
\renewenvironment{proof}{{\noindent\bf Proof.}}{ \hfill ~\qed}
\def\qed{\rule[0pt]{5pt}{5pt}\par\medskip}

\newcommand{\tp}{\mathsf{T}}		
\newcommand{\R}{\mathbb{R}}			
\newcommand{\ord}{\mathcal{O}}		
\newcommand{\eqt}{\buildrel\smash{\smash t}\over =}
\newcommand{\geqt}{\buildrel\smash{\smash t}\over \geq}
\newcommand{\gt}{\buildrel\smash{\smash t}\over >}

\DeclareMathOperator{\norml}{\mathcal{N}}		
\DeclareMathOperator{\ee}{\mathbb{E}}			
\DeclareMathOperator{\prob}{\mathbb{P}}			
\DeclareMathOperator{\vecc}{\mathbf{vec}}		
\DeclareMathOperator{\tr}{\mathbf{tr}}			

\newcommand{\bmat}[1]{\begin{bmatrix}#1\end{bmatrix}}
\newcommand{\probc}[2]{\prob\!\left(#1\,\middle\vert\,#2\right)}	
\newcommand{\eec}[2]{\ee\!\left(#1\,\middle\vert\,#2\right)}		
\newcommand{\probcs}[2]{\prob(#1\,\vert\,#2)}						
\newcommand{\eecs}[2]{\ee(#1\,\vert\,#2)}							

\usepackage{framed}
\newenvironment{red}{%
  
  \MakeFramed{\advance\hsize-\width \FrameRestore}}
  {\endMakeFramed}

\newcommand{\boxit}[1]{\vspace{5mm}\noindent
\fbox{\begin{minipage}{0.98\linewidth}#1\end{minipage}}\vspace{5mm}

}

\begin{document}

\if\MODE3
\title{Structural Results and Explicit Solution\\
		for Two-Player LQG Systems on a\\
		Finite Time Horizon}
\else
\title{Structural Results and Explicit Solution for Two-Player Partially Nested LQG Systems}
\fi

\if\MODE3\author{Laurent Lessard \and Ashutosh Nayyar}
\else\author{Laurent Lessard\footnotemark[1] \and
		Ashutosh Nayyar\footnotemark[2]}\fi
\note{To Appear, IEEE Conference on Decision and Control 2013}
\maketitle

\if\MODE3
\else
\footnotetext[1]{L.~Lessard is with the Department of Mechanical Engineering at the
University of California, Berkeley, CA~94720, USA. \texttt{lessard@berkeley.edu}}
\footnotetext[2]{A.~Nayyar is with the Department of Electrical Engineering and Computer Sciences, University of California, Berkeley, CA~94720, USA. \texttt{anayyar@berkeley.edu}}
\fi


\begin{abstract}
It is well-known that linear dynamical systems with Gaussian noise and quadratic cost (LQG) satisfy a separation principle. Finding the optimal controller amounts to solving separate dual problems; one for control and one for estimation. For the discrete-time finite-horizon case, each problem is a simple forward or backward recursion. In this paper, we consider a generalization of the LQG problem with two controllers and a partially nested information structure. Each controller is responsible for one of two system inputs, but has access to different subsets of the available measurements. Our paper has three main contributions. First, we prove a fundamental structural result: sufficient statistics for the controllers can be expressed as conditional means of the global state. Second, we give explicit state-space formulae for the optimal controller. These formulae are reminiscent of the classical LQG solution with dual forward and backward recursions, but with the important difference that they are intricately coupled. Lastly, we show how these recursions can be solved efficiently, with computational complexity comparable to that of the centralized problem.
\end{abstract}

\section{Introduction}\label{sec:intro}

With the advent of large systems operating on a global scale such as the internet or power networks, the past decade has seen a resurgence of interest in decentralized control. For such large systems, it is inevitable that some control decisions must be made using only local or partial information. Two natural questions that arise are:
\begin{enumerate}
\item Can the ever-growing information history be aggregated without compromising achievable performance? In other words, what are sufficient statistics for the decision-makers?
\item When and how can optimal decentralized policies be efficiently computed?
\end{enumerate}

\noindent In this paper, we give complete answers to the above questions for a fundamental decentralized control problem: the two-player partially nested LQG problem.

Briefly, our problem consists of two linear Gaussian systems with their own local controllers. The systems are coupled. System~1 affects System~2 through its state and input but not vice-versa, and the controller for System~1 shares its measurement with the controller for System~2 but not vice versa. A formal description of the problem is given in~Section~\ref{sec:prob_statement}. 

It is believed that decentralized control problems are likely hard in general~\cite{blondel,witsenhausen}. However, partially-nested LQG problems admit an optimal controller that is linear~\cite{hochu}. In decentralized control problems, partial nestedness is typically manifested in two ways: 
\begin{enumerate}
\item If a subsystem $i$ affects another subsystem $j$, then the controller at subsystem $i$ shares all information with the controller at subsystem $j$. In other words, the information flow obeys the \emph{sparsity constraints} of the dynamics.
\item If subsystem $i$ affects subsystem $j$ after some delay~$d$, then controller at subsystem $i$ shares its information with controller at subsystem $j$ with delay not exceeding $d$. In other words, the information flow obeys the \emph{delay constraints} of the dynamics.
\end{enumerate}
Several combinations of the above two manifestations of partial nestedness have been explored under state feedback assumptions.  State feedback problems have been investigated in~\cite{swigart10dp} under sparsity constraints, in~\cite{lamperski_delayed_recent} under delay constraints and in~\cite{lamperski_lessard} under a mixture of delay and sparsity constraints. A state feedback problem where partial nestedness was captured by a partial order on subsystems was investigated in \cite{shah10}.

The problem considered in this paper is an output feedback problem where controllers observe noisy measurements  of states. With only two controllers, it is perhaps the simplest output feedback problem.  However, as we shall see, having noisy measurements introduces a nontrivial coupling between estimation and control and complicates the solution significantly. An explicit solution to the continuous-time version of the two-player problem as well as an extension to the broadcast case appeared in~\cite{lessard_broadcast,lessard_allerton,lessard_acc}. These works use a spectral factorization approach that is completely different from the common information approach used herein. Furthermore, they solve the problem over an infinite time horizon, which makes the coupling between estimation and control simpler due to the steady-state assumption. A particular case of output-feedback partially nested LQG problem namely the one-step delayed sharing problem was investigated in~\cite{nonclassical}.

Our paper has three main contributions. In Section~\ref{sec:structural_results}, we find sufficient statistics for the two-player LQG problem. Our result relies on a common information based approach developed in~\cite{nayyar_thesis} and~\cite{nayyar}. In Section~\ref{sec:explicit}, we give an explicit state-space solution to the two-player problem using dynamic programming. Lastly, in Section~\ref{sec:efficient}, we show how to efficiently compute the solution to the two-player problem, and show that it can be done with computational effort comparable to that required for the centralized version of the problem. Namely, computational effort is proportional to the length of the time horizon.

\section{Notation}\label{sec:notation}

Real vectors and matrices are represented by lower- and upper-case letters respectively. Boldface symbols denote random vectors, and their non-boldface counterparts denote particular realizations. The probability density function of $\mathbf{x}$ evaluated at $x$ is denoted $\prob(\mathbf{x}=x)$, and conditional densities are written as $\probc{\mathbf{x}}{\mathbf{y}=y}$. We write
$
\mathbf{x} = \norml(\mu,\Sigma)
$
when $\mathbf{x}$ is normally distributed with mean $\mu$ and variance $\Sigma$. This paper considers stochastic processes in discrete time over a finite time interval $[0,T]$.
Time is indicated using subscripts, and we use the colon notation to denote ranges. For example:
$
x_{0:T-1} = \{ x_0, x_1,\dots, x_{T-1} \}
$.
In general, all symbols are time-varying. In an effort to present general results while keeping equations clear and concise, we introduce a new notation to represent a family of equations. We write
\[
\mathbf{x}_+ \eqt A\mathbf{x} + \mathbf{w}
\]
to mean that $\mathbf{x}_{t+1} = A_t \mathbf{x}_t + \mathbf{w}_t$ holds for $0 \le t \le T-1$. The subscript ``$+$'' indicates that the associated symbol is incremented to $t+1$. We similarly overload summations:
\[
\sum_t x^\tp Q x 
\quad\text{instead of writing}\quad \sum_{t=0}^{T-1} x_t^\tp Q_t x_t
\]
Any time we use $t$ above a binary relation or below a summation, it is implied that $0 \le t \le T-1$. The same time horizon $T$ is used throughout this paper.

We denote subvectors by using superscripts so that they are not confused with time indices. For submatrices, which require double indexing, we will interchangeably use superscripts and subscripts to minimize clutter. For example, we write $P_+^{21}$ and $A_{22}^\tp$ to avoid writing $P_{21,+}$ and $A^{22,\tp}$ respectively. We also introduce matrices $E_1$ and $E_2$ to aid in the manipulation of $2\times 2$ block matrices. We partition an identity matrix as $I = \bmat{E_1 & E_2}$ where the dimension of $E_i$ is inferred by context. For example, suppose $B \in \R^{(n_1+n_2)\times (m_1+m_2)}$, with the block-triangular structure given by
\[
B = \bmat{ B_{11} & 0 \\ B_{21} & B_{22} }\qquad \text{where }B_{ij}\in\R^{n_i\times m_j}
\]
then we may write $E_2^\tp B E_1 = B_{21}$ and $E_1^\tp B = B_{11} E_1^\tp$. However, the same symbol may vary in dimension depending on context. In this case, either $E_i\in\R^{(n_1+n_2)\times n_i}$ or $E_i\in\R^{(m_1+m_2)\times m_i}$ depending on whether the symbol multiplies $B$ on the left or right, respectively.

\section{Problem statement}\label{sec:prob_statement}

Consider two interconnected linear systems with the following state update and measurement equations.
\begin{equation}\label{eq:state_eqns}
\hspace{-2mm}\begin{alignedat}{2}
\bmat{ \mathbf{x}_+^1 \\ \mathbf{x}_+^2 } &\eqt\,&
	\addtolength{\arraycolsep}{-0.1em}\bmat{ A_{11} & 0 \\ A_{21} & A_{22} }\!
		&\bmat{ \mathbf{x}^1 \\ \mathbf{x}^2 } +
	\addtolength{\arraycolsep}{-0.1em}\bmat{ B_{11} & 0 \\ B_{21} & B_{22} }\!
		\bmat{ \mathbf{u}^1 \\ \mathbf{u}^2 } +
	\bmat{ \mathbf{w}^1 \\ \mathbf{w}^2 } \\
\bmat{ \mathbf{y}^1 \\  \mathbf{y}^2 } &\eqt\,&
	\addtolength{\arraycolsep}{-0.1em}\bmat{ C_{11} & 0 \\ C_{21} & C_{22} }\!
		&\bmat{  \mathbf{x}^1 \\  \mathbf{x}^2 } +
	\bmat{  \mathbf{v}^1 \\  \mathbf{v}^2 }
\end{alignedat}
\end{equation}
For brevity, we write the vector $(\mathbf{x}_t^1,\mathbf{x}_t^2)$ above as simply $\mathbf{x}_t$ and similarly for $\mathbf{y}_t$, $\mathbf{u}_t$, $\mathbf{w}_t$, $\mathbf{v}_t$.
The random vectors in the collection
\begin{equation}
\left\{
	 \mathbf{x}_0, \bmat{ \mathbf{w}_0\\ \mathbf{v}_0}, \dots, \bmat{ \mathbf{w}_{T-1}\\ \mathbf{v}_{T-1}}
\right\}
\end{equation}
are mutually independent and jointly Gaussian with the following known probability density functions.
\begin{equation}\label{eq:initial_conditions}
\begin{aligned}
\mathbf{x}_0 & = \norml(0, \Sigma_\textup{init}) \\
\bmat{\mathbf{w}\\\mathbf{v}} & \eqt \norml\biggl(0,\bmat{W & U^\tp \\ U & V}\biggr)
\end{aligned}
\end{equation}
There are two controllers, and the information available to each controller at time $t$ is
\begin{equation}\label{eq:information_pattern}
\begin{aligned}
	\hat{\mathbf{i}}_t &= \left\{ \mathbf{y}_{0:t-1}^1, \mathbf{u}_{0:t-1}^1 \right\} \\
	\mathbf{i}_t &= \left\{ \mathbf{y}_{0:t-1}^1, \mathbf{y}_{0:t-1}^2,  \mathbf{u}_{0:t-1}^1, \mathbf{u}_{0:t-1}^2 \right\}
\end{aligned}
\end{equation}
The controllers select actions according to control strategies $f^i :=(f^i_0,f^i_1,\ldots,f^i_{T-1})$ for $i=1,2$. That is,
\begin{equation}\label{eq:input_definition}
	\mathbf{u}_t^1 = f^1_t(\,\hat{\mathbf{i}}_t) 
	\,\,\,\text{and}\,\,\,
	\mathbf{u}_t^2 = f^2_t(\,\mathbf{i}_t)
	\quad\text{for }0\le t\le T-1
\end{equation}
The performance of control strategies $f^1,f^2$ is measured by the
finite horizon expected quadratic cost given by
\if\MODE3
\begin{equation}\label{eq:cost_2p}
\hat{\mathcal{J}}_0(f^1,f^2) =
\ee^{f^1\!,f^2} \biggl( \sum_t
\bmat{\mathbf{x}\\\mathbf{u}}^\tp \bmat{Q & S \\ S^\tp & R} \bmat{\mathbf{x}\\ \mathbf{u}}
+ \mathbf{x}_T^\tp P_\textup{final} \mathbf{x}_T \biggr)
\end{equation}
\else
\begin{multline}\label{eq:cost_2p}
\hat{\mathcal{J}}_0(f^1,f^2) = \\
\ee^{f^1\!,f^2} \biggl( \sum_t
\bmat{\mathbf{x}\\\mathbf{u}}^\tp \bmat{Q & S \\ S^\tp & R} \bmat{\mathbf{x}\\ \mathbf{u}}
+ \mathbf{x}_T^\tp P_\textup{final} \mathbf{x}_T \biggr)
\end{multline}
\fi
The expectation is taken with respect to the joint probability measure on $(\mathbf{x}_{0:T},  \mathbf{u}_{0:T-1})$ induced by the choice
of $f^1$ and $f^2$. 
We are interested in the following problem.
\vspace{-3mm}

\boxit{
\begin{problem}[Two-Player LQG]
  \label{prob:TPLQG}
  For the model~\eqref{eq:state_eqns}--\eqref{eq:input_definition}, find control strategies $f^1,f^2$
  that minimize the cost~\eqref{eq:cost_2p}.
\end{problem}
}
\vspace{-3mm}

A related and well-known problem is the centralized LQG problem. It is the special case of the two-player problem for which there is a single decision-maker.
\vspace{-3mm}

\boxit{
\begin{problem}[Centralized LQG]\label{prob:CLQG}
Consider the model~\eqref{eq:state_eqns}--\eqref{eq:initial_conditions}, where $A$, $B$, $C$ are no longer required to be block-lower-triangular.
Suppose $\mathbf{u}_t = f_t(\,\mathbf{i}_t)$, where $\mathbf{i}_t := (\mathbf{y}_{0:t-1},\mathbf{u}_{0:t-1})$, and our goal is to choose $f := f_{0:T-1}$ such that we minimize
\[
\mathcal{J}_0(f) =
\ee^{f} \biggl( \sum_t
\bmat{\mathbf{x}\\\mathbf{u}}^\tp \bmat{Q & S \\ S^\tp & R} \bmat{\mathbf{x}\\ \mathbf{u}}
+ \mathbf{x}_T^\tp P_\textup{final} \mathbf{x}_T \biggr)
\]
The expectation is with respect to the joint probability measure on
$(\mathbf{x}_{0:T}, \mathbf{u}_{0:T-1})$ induced by the choice of~$f$.
\end{problem}
}

In both Problem~\ref{prob:TPLQG} and Problem~\ref{prob:CLQG}, the sizes of the various matrices and vectors may also vary with time. It is assumed that $\Sigma_\text{init},P_\text{final}$, as well as the values of $A,B,C,Q,R,S,U,V,W$ for all $t$, are available to all decision-makers for any $t \ge 0$.
We also clarify that while we often call the decision-making agents \emph{players}, this is not a game. The players are cooperative and their strategies are to be jointly optimized.

\section{Structural results}\label{sec:structural_results}

In Problem \ref{prob:TPLQG}, the lower triangular nature of the state, control and observation matrices of \eqref{eq:state_eqns} implies that  Player~1's state and control actions affect Player~2's information but not vice versa. Further, any information available of Player~1 is also available to Player~2. Hence, Problem \ref{prob:TPLQG} is \emph{partially nested} and the optimal strategies for the two players are linear functions of their respective information histories~\hbox{\cite[Thm.~2]{hochu}}.

In this section, we show that the information histories can be aggregated into sufficient statistics. In Sections~\ref{sec:explicit} and~\ref{sec:efficient}, we will use this fact to derive a recursive finite-memory implementation of the optimal controller.
We start with a well-known structural result for the centralized LQG problem (Problem \ref{prob:CLQG}).

\begin{lem}\label{lem:lqg_str}
In Problem \ref{prob:CLQG}, there exists an optimal control strategy of the form $u \eqt K z$
where $z_t := \eec{\mathbf{x}_t}{\mathbf{i}_t=i_t}$, and $K_{0:T-1}$ are fixed matrices of appropriate dimensions.
\end{lem}

We will also make use of some properties of conditional expectations, which apply because of the nested information. We state the result as a lemma.

\begin{lem}\label{lem:smoothing}
Suppose $\hat{\mathbf{i}}$ and  $\mathbf{i}$ are information sets that satisfy $\hat{\mathbf{i}} \subset \mathbf{i}$. Define conditional estimates
$z := \eec{\mathbf{x}}{\mathbf{i}=i}$ and
$\hat z := \eecs{\mathbf{x}}{\hat{\mathbf{i}}=\hat i\,}$. Then
\begin{enumerate}
\item [(i)] $\eecs{\hat{\,\mathbf{i}}}{\mathbf{i} = i\,} = \hat i$
\qquad (since $\hat i \subset i$)
\item [(ii)] $\eecs{\mathbf{z}}{\hat{\mathbf{i}}=\hat i\,} = \hat z $
\!\qquad (smoothing property)
\end{enumerate}
\end{lem}

\subsection{Structural result for Player~2}\label{sec:first_str}

We now turn our attention to Problem~\ref{prob:TPLQG}. 
Consider any arbitrary linear strategy for Player~1. Thus, Player~1's control actions are of the form
\begin{equation} \label{eq:assumed1}
 u^1 \eqt G \,\hat i
\end{equation}
where $G_{0:T-1}$ are fixed matrices of appropriate dimensions and $\hat i_t$ is the realization of Player~1's information. Given this strategy for Player~1, we want to find the optimal strategy for Player~2.

In the next two results, we show that once Player~1's strategy is fixed, finding the optimal strategy for Player~2 amounts to solving a centralized LQG problem. Thus we may apply the structural result presented in Lemma~\ref{lem:lqg_str}.

\begin{lem} \label{lem:newstate1}
Consider Problem~\ref{prob:TPLQG}, and assume any fixed strategy for Player~1 given by \eqref{eq:assumed1}. Define $\bar{\mathbf{x}}_t$ as follows.
\[
\bar{\mathbf{x}}_t := \bmat{\mathbf{x}_t \\ \hat{\mathbf{i}}_t}, \bar{\mathbf{y}}_t := \bmat{\mathbf{y}_t \\ \hat{\mathbf{i}}_t}
\quad\text{for}\quad
0\le t\le T
\]
Then, the following statements are true.
\begin{enumerate}
\item [(i)] There exist matrices $\bar A_t, \bar B_t, \bar C_t, \bar D_t$ such that 
\begin{align*}
\bar{\mathbf{x}}_0 &= \norml(0,\Sigma_\textup{init}) \\
 \bar{\mathbf{x}}_+ &\eqt \bar A \bar{\mathbf{x}} + \bar B \mathbf{u}^2
 		+ \bar D \bmat{ \mathbf{w} \\ \mathbf{v} } \\
 \bar{\mathbf{y}} &\eqt \bar C \bar{\mathbf{x}} + \mathbf{v}
\end{align*}
\item [(ii)] There exist matrices $\bar Q_t, \bar R_t, \bar S_t,\bar P_\textup{final}$ such that the total expected cost can be written as
\begin{align*}
\ee \left( \sum_t
\bmat{\bar{\mathbf{x}}\\\mathbf{u}^2}^\tp
\bmat{ \bar Q & \bar S \\ \bar S^\tp & \bar R }
\bmat{\bar{\mathbf{x}}\\\mathbf{u}^2}
+ \bar{\mathbf{x}}_T^\tp \bar P_\textup{final} \bar{\mathbf{x}}_T \right)
\end{align*}
\end{enumerate}
\end{lem}
\begin{proof}
The proof follows from the definition of $\bar{\mathbf{x}}_t$, the state, observation, and cost equations of Problem~\ref{prob:TPLQG}, and the fixed strategy for Player~1 given by~\eqref{eq:assumed1}. 
\end{proof}

\begin{thm} \label{thm:first_str}
Consider Problem~\ref{prob:TPLQG}. For any choice of Player~1's strategy, the optimal strategy for Player~2 has the structure
\begin{equation}
 u^2 \eqt 
 H^1\, \hat i + H^2 z
\end{equation}
where $z_t := \eec{\mathbf{x}_t}{\mathbf{i}_t=i_t}$ and $\hat i_t \subset i_t$ is the realization of Player~1's information. Further, $z_t$ has a linear update equation that does not depend on the choice of Player~1's strategy. This linear update equation is the standard Kalman filter, given explicitly in~\eqref{eq:centr_opt}--\eqref{eq:sol_Lt}.
\end{thm}
\begin{proof}
Lemma \ref{lem:newstate1} implies that when Player~1's strategy is fixed, the optimization problem for Player~2 is an instance of the centralized LQG problem (Problem~\ref{prob:CLQG}) with $\bar{\mathbf{x}}_t$ as the state of the linear system, $\mathbf{y}_t$ as the observation, and $\mathbf{u}^2_t$ as the control action. Therefore, by Lemma~\ref{lem:lqg_str}, the optimal strategy for Player~2 is of the form $u^2_t = {H_t} \eec{\bar{\mathbf{x}}_t}{\mathbf{i}_t=i_t}$ for some matrix~$H_t$. Further, it follows from Lemma~\ref{lem:smoothing} that 
 \begin{align}
 &\eec{\bar{\mathbf{x}}_t}{\mathbf{i}_t=i_t}
 	= \bmat{\eec{\mathbf{x}_t}{\mathbf{i}_t=i_t} \\[1mm]
 		\eecs{\,\hat{\mathbf{i}}_t }{ \mathbf{i}_t=i_t}}
 	= \bmat{z_t \\[1mm] \hat{i}_t} \label{eq:ANthmproof1}
 \end{align}
 Therefore, the optimal strategy for Player~2 is of the form
 $u^2_t = {H_t} \eec{\bar{\mathbf{x}}_t}{\mathbf{i}_t=i_t}
 = H^1_t\,\hat{i}_t + H^{2}_t z_t$, as required.
\end{proof}

\subsection{Joint structural result}\label{sec:second_str}
We may rewrite the result of Theorem~\ref{thm:first_str} as
\begin{equation} \label{eq:coord1}
 u^2 \eqt \tilde u^2 + H^2 z
\end{equation}
where $\tilde u^{2}_t$ is  a linear function of Player~1's information. Note that $\tilde u^2_t$ and $u^1_t$ are linear functions of the same information. In order to further characterize the structure of optimal strategies, we consider a coordinated system where  a coordinator knows the common information among the players (that is,~$\hat i_t$) and selects both $\tilde u^{2}_t$ and $u^1_t$ based on this common information. Once the coordinator selects $\tilde u^{2}_t$, Player~2's control action is $u^2_t = H^{2}_t z_t + \tilde u^{2}_t$, for some $H^2_t$. It is clear that any strategy of the form~\eqref{eq:coord1} can be implemented in the coordinated system.

Given an arbitrary choice of $H^2_t$, we want to find the optimal strategy for the coordinator. As in~Lemma~\ref{lem:newstate1}, this can be formulated as a centralized LQG problem.
\begin{lem} \label{lem:newstate2}
Consider Problem~\ref{prob:TPLQG} where $u^2_t$ is given by~\eqref{eq:coord1}, and assume any fixed choice of $H^2_t$. Define~$\tilde{\mathbf{x}}_t$ as follows.
\[
\tilde{\mathbf{x}}_t := \bmat{\mathbf{x}_t \\ \mathbf{z}_t}
\quad\text{for}\quad
0\le t\le T
\]
Then, the following statements are true.
\begin{enumerate}
\item [(i)] There exist matrices $\tilde A_t$, $\tilde B_t$, $\tilde D_t$, and $\tilde \Sigma_\textup{init}$ such that 
\begin{align*}
\tilde{\mathbf{x}}_0 &= \norml(0,\tilde\Sigma_\textup{init}) \\
 \tilde{\mathbf{x}}_+ &\eqt \tilde A \tilde{\mathbf{x}} + \tilde B \bmat{\mathbf{u}^1 \\ \tilde{\mathbf{u}}^2} + \tilde D \bmat{ \mathbf{w} \\ \mathbf{v} } \\
 \mathbf{y}^1 &\eqt C_{11}\mathbf{x}^1 + \mathbf{v}^1
\end{align*}
\item [(ii)] There exist matrices $\Omega_t$ and $\Omega_\textup{final}$ such that the total expected cost can be written as
\begin{align*}
\ee \Biggl( \sum_t\bmat{\tilde{\mathbf{x}} \\ \mathbf{u}^1 \\ \tilde{\mathbf{u}}^2 }^\tp
 \!\Omega \bmat{\tilde{\mathbf{x}} \\ \mathbf{u}^1 \\ \tilde{\mathbf{u}}^2 }
+ \tilde{\mathbf{x}}_T^\tp\, \Omega_\textup{final}\, \tilde{\mathbf{x}}_T
 \Biggr)
\end{align*}
\end{enumerate}
\end{lem}
\begin{proof}
The proof follows from the definition of $\tilde{\mathbf{x}}_t$,  the state, observation, and cost equations of Problem~\ref{prob:TPLQG}, and the fact that $z_t$ has a linear update equation. This linear update equation is the standard Kalman filter, given explicitly in~\eqref{eq:centr_opt}--\eqref{eq:sol_Lt}.
\end{proof}

\begin{thm}\label{thm:second_str}
The optimal strategies for the two players in Problem~\ref{prob:TPLQG} are of the form
\begin{align}
u^1_t = G_t \hat z_t \qquad  u^{2}_t = H^1_t \hat z_t + H^2_t z_t
\end{align}
where $\hat z_t := \eecs{\mathbf{x}_t }{ \hat{\mathbf{i}}_t=\hat i_t}$ and
 $z_t := \eecs{\mathbf{x}_t }{ {\mathbf{i}}_t= i_t}$.
\end{thm}
\begin{proof}
Lemma \ref{lem:newstate2} implies that when $H^2_t$ is fixed in~\eqref{eq:coord1}, the coordinator's optimization problem is an instance of Problem~\ref{prob:CLQG} with $\tilde{\mathbf{x}}_t$ as the state of the linear system and $( \mathbf{u}^1_t, \tilde{\mathbf{u}}^2_t )$
as the control action. By Lemma~\ref{lem:lqg_str}, we obtain 
\begin{align}
\bmat{ u^1_t \\ \tilde u^{2}_t} = \tilde H_t \eecs{\tilde{\mathbf{x}}_t}{ \hat{\mathbf{i}}_t= \hat i_t}
= \tilde H_t \eec{\bmat{\mathbf{x}_t \\ \mathbf{z}_t}}{\hat{\mathbf{i}}_t=\hat i_t}
\end{align}
for some matrix $\tilde H_t$. The first component of the expectation is simply $\hat z_t$, and the second component is also $\hat z_t$ by Lemma~\ref{lem:smoothing}. Therefore,~$u^1_t$ and~$\tilde u^2_t$ are linear functions of $\hat z_t$, Player~1's estimate.
\end{proof}
Theorem~\ref{thm:second_str} shows that for Problem~\ref{prob:TPLQG}, a sufficient statistic is the set of conditional means $(\hat z_t, z_t)$.  Note that for a given realization of $\mathbf{i}_t$, Player~2's estimate $z_t$ does not depend on players' strategies. However, for a given realization of $\hat{\mathbf{i}}_t$, Player~1's estimate $\hat z_t$ depends on the choice of matrices $H^2_{0:t-1}$.

\subsection{Extension to the POMDP case}\label{sec:discussion}

In the centralized LQG problem (Problem~\ref{prob:CLQG}), the fact that the optimal control action is a linear function of the conditional mean of the state is a consequence of the LQG assumptions. If state update and measurement equations are nonlinear with non-Gaussian noise, the centralized problem becomes a partially observable Markov decision process (POMDP). For POMDPs, optimal actions are functions of the \emph{conditional probability density of the state} and not just the conditional mean. In other words:
\begin{align*}
\text{LQG: } & & u_t &= K_t z_t
	 & &\text{where } & z_t &= \eecs{\mathbf{x}_t}{\mathbf{i}_t = i_t}\\
\text{POMDP: } & & u_t &= \phi_t( \pi_t )
	 & &\text{where } & \pi_t &= \probcs{\mathbf{x}_t}{\mathbf{i}_t = i_t}
\end{align*}
where $\phi_t$ is a (possibly) nonlinear function, and $\pi_t$ is called the \emph{belief state}. Note that $\pi_t$ is a probability density function while $z_t$ is simply a real vector.

For the two-player problem (Problem~\ref{prob:TPLQG}), the simple form of the structural results in Theorems~\ref{thm:first_str} and~\ref{thm:second_str} is a consequence of the triangular information structure and the LQG assumptions. We now investigate the corresponding two-player structural result for POMDPs.

A straightforward extension of Lemma~\ref{lem:newstate1} shows that for any choice of Player~1's strategy, Player~2's optimization problem is a POMDP and the optimal strategy has the form
$ u^2_t = \gamma_t(\pi_t, \hat{i}_t) $
where $\gamma_t$ is a (possibly) nonlinear function.
%
%
Because Player~2's optimal policy is no longer linear, the coordinator's problem of Section~\ref{sec:second_str} is more complicated. The coordinator is now required to select a control action for Player~1 and a \emph{function that maps Player~2's belief to Player~2's action}. The coordinator's problem can be viewed as a POMDP with $(\mathbf{x}_t, \boldsymbol\pi_t)$ as the state. Therefore, the associated structural result involves a belief on the pair $(\mathbf{x}_t, \boldsymbol\pi_t)$. Not only is the coordinator required to keep a belief on the state, it must also keep a \emph{belief on Player~2's belief on the state}.

As shown in Section~\ref{sec:structural_results}, the structures of the optimal strategies for the centralized and two-player problems in the LQG case are of comparable complexity. However, this is not the case for the nonlinear, non-Gaussian versions of these problems. The centralized case requires maintaining a belief on the system state, while the two-player case requires maintaining a belief on a belief. This is substantially more complicated object.


\section{Explicit solution}\label{sec:explicit}

In this section, we use the structural results of Section~\ref{sec:structural_results} to derive an explicit and efficiently computable state-space realization for the optimal controller for Problem~\ref{prob:TPLQG}.

To ensure a unique optimal controller with a recursively computable structure, we make some additional mild assumptions, which we list below.

\paragraph{Main assumptions.}
We assume the following.
\begin{align}\label{ass1}
\bmat{W & U^\tp \\ U & V} &\geqt 0,&
\Sigma_\textup{init} &\ge 0,&
&\text{and}&
\quad V &\gt 0 \\
\bmat{Q & S \\ S^\tp & R} &\geqt 0,&
P_\textup{final} &\ge 0,&
&\text{and}&
\quad R &\gt 0 \label{ass2}
\end{align}
The assumptions that $V_t > 0$ and $R_t > 0$ are made for simplicity and can generally be relaxed. For example, it is only required that $C_t \Sigma_t C_t^\tp + V_t > 0$, so as long as this holds, we can have $V_t \ge 0$.

The well-known solution to the centralized LQG problem (Problem~\ref{prob:CLQG}) in given in the following lemma.

\begin{lem}\label{lem:centralized_explicit}
Consider Problem~\ref{prob:CLQG} and suppose the main assumptions~\eqref{ass1}--\eqref{ass2} hold.
The optimal policy is
\begin{equation}\label{eq:centr_opt}
\begin{aligned}
z_0 &= 0\\
z_+ &\eqt A z + B u - L (y - Cz) \\
u &\eqt K z
\end{aligned}
\end{equation}
where $L_{0:T-1}$ satisfies the forward recursion
\begin{equation}\label{eq:sol_Lt}
\begin{aligned}
\Sigma_0 &= \Sigma_\textup{init} \\
\Sigma_+ &\eqt A\Sigma A^\tp + L(C\Sigma A^\tp + U) + W \\
L &\eqt -(A\Sigma C^\tp + U^\tp) (C\Sigma C^\tp + V)^{-1}
\end{aligned}
\end{equation}
and $K_{0:T-1}$ satisfies the backward recursion
\begin{equation}\label{eq:sol_Kt}
\begin{aligned}
P_T &= P_\textup{final} \\
P &\eqt A^\tp P_+ A + (A^\tp P_+ B + S) K + Q \\
K &\eqt -(B^\tp P_+ B + R)^{-1} (B^\tp P_+ A + S^\tp)
\end{aligned}
\end{equation}
For every $t$, the belief state has the distribution
\begin{equation}\label{eq:belief_centralized}
\probc{\mathbf{x}_t}{\mathbf{i}_t = i_t} = \norml(z_t, \Sigma_t)
\end{equation}
and the optimal average cost is given by
\if\MODE3
\begin{equation}
\label{eq:cost_centralized}
\mathcal{J}_0 =
\tr(P_0 \Sigma_\textup{init})
+\sum_t \Bigl(\, \tr(P_+ W) + \tr\bigl[ \Sigma K^\tp(B^\tp P_+ B + R) K\bigr] \,\Bigr)
\end{equation}
\else
\begin{multline}
\label{eq:cost_centralized}
\mathcal{J}_0 =
 \tr(P_0 \Sigma_\textup{init}) \\
+\sum_t \Bigl(\, \tr(P_+ W) + \tr\bigl[ \Sigma K^\tp(B^\tp P_+ B + R) K\bigr] \,\Bigr)
\end{multline}
\fi
\end{lem}
\begin{proof}
See for example \cite{kumar_varaiya} or \cite{astrom}.
\end{proof}

Note that Lemma~\ref{lem:centralized_explicit} holds in great generality. All system, cost, and covariance matrices may vary with time. The above formulae hold even in the case where the dimensions of the matrices are different at every timestep.

\if\MODE3\else\newpage\fi
The main result of this section is a state-space solution to Problem~\ref{prob:TPLQG}, the two-player problem. The result, given below in Theorem~\ref{thm:tpof_explicit}, is similar in structure and generality to Lemma~\ref{lem:centralized_explicit}, but with one important difference. In Lemma~\ref{lem:centralized_explicit}, the gains~$K_{0:T-1}$ and~$L_{0:T-1}$ can be computed separately using different recursions. Thus, Lemma~\ref{lem:centralized_explicit} provides both a solution and a recipe for its construction. In contrast, the recursions for~$\hat K_{0:T-1}$ and~$\hat L_{0:T-1}$ found in Theorem~\ref{thm:tpof_explicit} are coupled. Therefore, Theorem~\ref{thm:tpof_explicit} provides an implicitly defined solution, but no obvious construction method. In Section~\ref{sec:efficient}, we make our solution explicit by showing how the various gains described in Theorem~\ref{thm:tpof_explicit} can be efficiently computed.

\begin{thm}\label{thm:tpof_explicit}
Consider Problem~\ref{prob:TPLQG} and suppose the main assumptions~\eqref{ass1}--\eqref{ass2} hold. The optimal policy is
\begin{align}\label{eq:tpof_opt_1}
&\begin{aligned}
\hat z_0 &= 0 \\
\hat z_+ &\eqt A\hat z + B\hat u - \hat L( y - C\hat z ) \\
\hat u &\eqt K \hat z
\end{aligned} \\[1mm]
&\begin{aligned}\label{eq:tpof_opt_2}
z_0 &= 0 \\
z_+ &\eqt Az + Bu - L( y - Cz ) \\
u &\eqt K \hat z + \hat K ( z - \hat z )
\end{aligned}
\end{align}
where $L_{0:T-1}$ and $K_{0:T-1}$ satisfy~\eqref{eq:sol_Lt}--\eqref{eq:sol_Kt}. If we define $\hat A :\eqt A+B\hat K + \hat L C$, then $\hat L_{0:T-1}$ satisfies the recursion
\if\MODE3
\begin{equation}\label{sigma_hat}
\begin{aligned}
\hat \Sigma_0 &= \Sigma_\textup{init} \\
\hat \Sigma_+ &\eqt \Sigma_+ + \hat A (\hat \Sigma - \Sigma) \hat A^\tp
	+ (\hat L - L ) (C\Sigma C^\tp + V) (\hat L - L )^\tp  \\
\hat L &\eqt -\bigl(A\hat\Sigma C^\tp + U^\tp + B\hat K (\hat\Sigma-\Sigma)C^\tp\bigr)
	 E_1 (C_{11} \hat\Sigma^{11} C_{11}^\tp + V_{11})^{-1}E_1^\tp
\end{aligned}
\end{equation}
\else
\begin{equation}\label{sigma_hat}
\begin{aligned}
\hat \Sigma_0 &= \Sigma_\textup{init} \\
\hat \Sigma_+ &\eqt \Sigma_+ + \hat A (\hat \Sigma - \Sigma) \hat A^\tp \\
	&\hspace{12mm} + (\hat L - L ) (C\Sigma C^\tp + V) (\hat L - L )^\tp  \\
\hat L &\eqt -\bigl(A\hat\Sigma C^\tp + U^\tp + B\hat K (\hat\Sigma-\Sigma)C^\tp\bigr) \\
	&\hspace{12mm}\times E_1 (C_{11} \hat\Sigma^{11} C_{11}^\tp + V_{11})^{-1}E_1^\tp
\end{aligned}
\end{equation}
\fi
and $\hat K_{0:T-1}$ satisfies the recursion
\if\MODE3
\begin{equation} \label{P_hat}
\begin{aligned}
\hat P_T &= P_\textup{final} \\
\hat P &\eqt P + \hat A^\tp (\hat P_+ - P_+) \hat A
	+ (\hat K - K )^\tp (B^\tp P_+ B + R) (\hat K - K ) \\
\hat K &\eqt - E_2 (B_{22}^\tp \hat P^{22}_+ B_{22} + R_{22})^{-1}E_2^\tp 
	 \bigl( B^\tp \hat P_+ A + S^\tp + B^\tp (\hat P_+ - P_+)\hat L C \bigr)
\end{aligned}
\end{equation}
\else
\begin{equation} \label{P_hat}
\begin{aligned}
\hat P_T &= P_\textup{final} \\
\hat P &\eqt P + \hat A^\tp (\hat P_+ - P_+) \hat A \\
	&\hspace{12mm} + (\hat K - K )^\tp (B^\tp P_+ B + R) (\hat K - K ) \\
\hat K &\eqt - E_2 (B_{22}^\tp \hat P^{22}_+ B_{22} + R_{22})^{-1}E_2^\tp \\
	&\hspace{12mm} \times \bigl( B^\tp \hat P_+ A + S^\tp + B^\tp (\hat P_+ - P_+)\hat L C \bigr)
\end{aligned}
\end{equation}
\fi
where $E_i$ are matrices defined in Section~\ref{sec:notation}.
For every $t$, the belief states have the distributions
\begin{equation}\label{eq:belief}
\begin{aligned}
\probcs{\mathbf{x}_t}{\,\hat{\mathbf{i}}_t = \hat i_t} &= \norml(\hat z_t, \hat \Sigma_t) \\
\probc{\mathbf{x}_t}{\,\mathbf{i}_t = i_t} &= \norml(z_t, \Sigma_t)
\end{aligned}
\end{equation}
and the optimal average cost is given by
\if\MODE3
\begin{multline}
\label{eq:cost_twoplayer}
\hat{\mathcal{J}}_0 =
\tr(P_0 \Sigma_\textup{init})
+\sum_t \Bigl(\, \tr(P_+ W) + \tr\bigl[ \Sigma K^\tp(B^\tp P_+ B + R) K\bigr] \\
+ \tr\bigl[ (\hat\Sigma-\Sigma) (\hat K-K)^\tp\! (B^\tp P_+ B + R)(\hat K-K)\bigr] \Bigr)
\end{multline}
\else
\begin{multline}
\label{eq:cost_twoplayer}
\hat{\mathcal{J}}_0 =
\tr(P_0 \Sigma_\textup{init}) \\[3pt]
+\sum_t \Bigl(\, \tr(P_+ W) + \tr\bigl[ \Sigma K^\tp(B^\tp P_+ B + R) K\bigr] \\[-3pt]
+ \tr\bigl[ (\hat\Sigma-\Sigma) (\hat K-K)^\tp\! (B^\tp P_+ B + R)(\hat K-K)\bigr] \Bigr)
\end{multline}
\fi
\end{thm}

\begin{proof}
Note that $\hat LE_2 = 0$ and $E_1^\tp \hat K=0$. That is, the second block-column of $\hat L$ and the first block-row of $\hat K$ are zero. The required triangular structure is therefore satisfied because $\hat z_+$ only depends on~$y^1$ in~\eqref{eq:tpof_opt_1} and~$u^1$ only depends on~$\hat z$ in~\eqref{eq:tpof_opt_2}.

Since Player 2 observes all measurements and control actions, its estimate of the state is the standard Kalman filter. Therefore, $z_t$ evolves according to \eqref{eq:tpof_opt_2} where $L_{0:T-1}$ satisfies~\eqref{eq:sol_Lt}.

The result of Theorem~\ref{thm:second_str} implies that the optimal control vector can be expressed as
$
u \eqt \tilde{u} + \hat K z
$
where $\tilde{u}_t$ is chosen by the coordinator and  $\hat K_t$ is a matrix whose first block-row is zero, so $E_1^\tp \hat K_t = 0$. Our first step will be to fix  $\hat K_t$ for all $t$ and to optimize for the coordinator's strategy. We begin by computing 
$
\hat z_t = \eecs{\mathbf{x}_t}{\hat{\mathbf{i}}_t=\hat i_t}.
$
To this end, we construct an equivalent centralized problem and appeal once again to Lemma~\ref{lem:centralized_explicit}. By Lemma~\ref{lem:smoothing}, we have 
$ \hat z_t = \eecs{\mathbf{z}_t}{\hat{\mathbf{i}}_t=\hat i_t}$
so we may estimate $\mathbf{x}_t$ by estimating $\mathbf{z}_t$ instead. Substituting the definitions for $\mathbf{u}_t$ and $\mathbf{y}_t$ into~\eqref{eq:state_eqns} and \eqref{eq:tpof_opt_2}, we obtain the state equations
\begin{align*}
\bmat{ \mathbf{z}_+ \\ \mathbf{e}_+} &\eqt
	\bmat{A+B\hat K & -LC \\ 0 & A+LC} \bmat{ \mathbf{z} \\ \mathbf{e} }
		+ \bmat{B \\ 0} \tilde{\mathbf{u}} + \bmat{-L\mathbf{v} \\ \mathbf{w}+L\mathbf{v}} \\
\mathbf{y}^1 &\eqt \bmat{ E_1^\tp C & E_1^\tp C } \bmat{ \mathbf{z} \\ \mathbf{e} } + \mathbf{v}^1
\end{align*}
where we have defined the error signal $\mathbf{e} :\eqt \mathbf{x} - \mathbf{z}$. Apply Lemma~\ref{lem:centralized_explicit} to compute the $\Sigma$-recursion. A straightforward induction argument shows that the covariance and gain matrices that satisfy~\eqref{eq:sol_Lt} at time $t$ are given by
\[
\bmat{\hat \Sigma_t - \Sigma_t & 0 \\ 0 & \Sigma_t}
\quad\text{and}\quad
\bmat{ \hat L_t E_1 \\ 0 }
\]
where $\hat \Sigma_t$ and $\hat L_t$ satisfy~\eqref{sigma_hat}. 
 Computing the estimation equations~\eqref{eq:centr_opt}, we find that the estimate of ${\mathbf{e}}_t$ is $0$, and
\begin{equation}\label{eq:zhat_est}
\hat{\mathbf{z}}_+ \eqt \hat A \hat{\mathbf{z}} + B\tilde{\mathbf{u}} -\hat L \mathbf{y}
\end{equation}
where $\hat A$ is defined in the theorem statement.
State and input split into conditionally independent parts
\[
\bmat{\mathbf{x} \\ \mathbf{u}} \eqt \bmat{\mathbf{z} \\ \tilde{\mathbf{u}} + \hat K\mathbf{z}} + \bmat{\mathbf{e} \\ 0}
\]
so the only relevant part of the cost~\eqref{eq:cost_2p} is
\[ 
\ee\biggl( \sum_t \bmat{\mathbf{z} \\ \tilde{\mathbf{u}} + \hat K\mathbf{z}}^\tp 
\bmat{Q & S \\ S^\tp & R} 
\bmat{\mathbf{z} \\ \tilde{\mathbf{u}} + \hat K\mathbf{z}}
+ \mathbf{z}_T^\tp P_\textup{final} \mathbf{z}_T \biggr)
\]
Applying the $P$-recursion~\eqref{eq:sol_Kt} from Lemma~\ref{lem:centralized_explicit} to solve for $\tilde u$, we find after some algebra that
\begin{equation}\label{utilde}
\tilde u \eqt (K - \hat K)\hat z
\end{equation}
where $K_{0:T-1}$ is the centralized gain given by~\eqref{eq:sol_Kt}. Substituting~\eqref{utilde} into~\eqref{eq:zhat_est}, we recover the desired form for Player~1's estimator~\eqref{eq:tpof_opt_1}.

We have shown thus far that for any fixed $\hat K_t$, the optimal Player 1 estimator has the form specified in \eqref{eq:tpof_opt_1}, where~\eqref{sigma_hat}, and~\eqref{eq:belief} hold and the coordinator's action is given by~\eqref{utilde}.
Note that  since the first block-row of $\hat K_t$ is zero, it follows from \eqref{utilde} that  $ u^1 \eqt E_1^\tp K \hat z$ where $K_{0:T-1}$ is the centralized gain. Therefore, the optimal strategies for the two players have the following structure:
\begin{align}\label{eq:pbpo1}
&\begin{aligned}
\hat z_0 &= 0 \\
\hat z_+ &\eqt (A+BK)\hat z  - \hat L( y - C\hat z ) \\
u^1 &\eqt E_1^\tp K \hat z
\end{aligned} \\[1mm]
\label{eq:pbpo2}
&\begin{aligned}
z_0 &= 0 \\
z_+ &\eqt Az + Bu - L( y - Cz ) \\
u^2 &\eqt E_2^\tp K \hat z + E_2^\tp\hat K ( z - \hat z )
\end{aligned}
\end{align}
for some $\hat L_{0:T-1}$ and $\hat K_{0:T-1}$.
Because our problem is partially nested, a strategy of the form~\eqref{eq:pbpo1}--\eqref{eq:pbpo2} is globally optimal if and only if it is person by person optimal.

For a given $\hat K_{0:T-1}$, Player 1's strategy of the form \eqref{eq:pbpo1} will coincide with the coordinator's best response to  $\hat K_{0:T-1}$ if $\hat L_{0:T-1}$ satisfies~\eqref{sigma_hat}. Therefore, a strategy of the form \eqref{eq:pbpo1} with $\hat L_{0:T-1}$ satisfying~\eqref{sigma_hat} must be Player~1's best response to $\hat K_{0:T-1}$.
For a given choice of  $\hat L_{0:T-1}$, we now seek the best response of Player 2 of the form in \eqref{eq:pbpo2}. The combined control vector of the two players can be written as
$
u \eqt K \hat z + E_2\bar u,
$
where we allow $\bar u$ to be a function of player 2's entire information.
Gathering the state equations~\eqref{eq:state_eqns} and the estimator equations~\eqref{eq:tpof_opt_1}--\eqref{eq:tpof_opt_2}, we obtain
\begin{align*}
\bmat{ {\mathbf{x}}_+ \\ \hat{\mathbf{e}}_+ } &\eqt
\advance\arraycolsep-2pt
\bmat{A\!+\!BK & -BK \\ 0 & A\!+\!\hat LC}
\! \bmat{ {\mathbf{x}} \\ \hat{\mathbf{e}} }
+ \bmat{BE_2 \\ BE_2}\!\mathbf{\bar u} +
\bmat{ \mathbf{w} \\ \mathbf{w}+\hat L\mathbf{v} } \\
\mathbf{y} &\eqt \bmat{C & 0} \bmat{\mathbf{x} \\ \hat{\mathbf{e}}} + \mathbf{v}
\end{align*}
where we have defined the error signal~$\hat{\mathbf{e}} :\eqt \mathbf{x}-\hat{\mathbf{z}}$. The cost~\eqref{eq:cost_2p} is given by
\[
\ee\biggl( \sum_t \bmat{\mathbf{x} \\ K\hat{\mathbf{z}} + E_2\bar{\mathbf{u}}}^\tp \!
\advance\arraycolsep-2pt
\bmat{Q & S \\ S^\tp & R}\! 
\bmat{\mathbf{x} \\ K\hat{\mathbf{z}} + E_2\bar{\mathbf{u}}}
+ \mathbf{x}_T^\tp P_\textup{final} \mathbf{x}_T \!\biggr)
\]
where the correct coordinates can be obtained by by substituting $\hat{\mathbf{z}}\eqt \mathbf{x}-\hat{\mathbf{e}}$.
Now apply Lemma~\ref{lem:centralized_explicit} to compute the $P$-recursion. A straightforward induction argument shows that the cost-to-go and gain matrices that satisfy~\eqref{eq:sol_Kt} at time~$t$ are given by
\[
\bmat{ P_t & 0 \\ 0 & \hat P_t - P_t}
\quad\text{and}\quad
\bmat{ 0 & E_2^\tp \hat K_t}
\]
where $\hat P_t$ and $\hat K_t$ satisfy~\eqref{P_hat}. It follows from Lemma~\ref{lem:centralized_explicit} that the optimal input is
\begin{align*}
\bar u_t &= \bmat{ 0 & E_2^\tp \hat K_t }
\bmat{ \eecs{\mathbf{x}_t}{\mathbf{i}_t = i_t} \\[1mm]
		\eecs{\hat{\mathbf{e}}_t}{\mathbf{i}_t = i_t} } 
= E_2^\tp \hat K_t(z_t - \hat z_t)
\end{align*}
Despite allowing $\bar u_t$ to depend on the full measurement history $i_t$, we find that it only depends on $(z_t-\hat z_t)$, so that Player 2's best response is of the form \eqref{eq:pbpo2}.

Thus, $\hat K_{0:T-1}, \hat L_{0:T-1}$ satisfying \eqref{sigma_hat} and \eqref{P_hat} constitute a person-by-person optimal, and consequently, a globally optimal solution of our problem.
\end{proof}

\begin{cor}[Nonzero initial state]
If the initial state for Problems~\ref{prob:TPLQG} and~\ref{prob:CLQG} is changed to
$\mathbf{x}_0 = \norml(\mu_\textup{init}, \Sigma_\textup{init})$,
then Lemma~\ref{lem:centralized_explicit} and Theorem~\ref{thm:tpof_explicit} change as follows.
\begin{itemize}
 \item[(i)] Estimators initialized at $z_0 = \hat z_0 = \mu_\textup{init}$.
\item[(ii)] Costs $\mathcal{J}_0$ and $\hat{\mathcal{J}}_0$ each increased by $\mu_\textup{init}^\tp P_0 \mu_\textup{init}$.
\end{itemize}
\end{cor}
\begin{proof}
Add a pre-initial timestep $\mathbf{x}_{-1} = \norml(0,I)$ with trivial dynamics $A_{-1}=C_{-1}=I$, $B_{-1}=0$ and apply the zero-initial-mean result to the augmented system.
\end{proof}

\section{Efficient computation}\label{sec:efficient}
Theorem~\ref{thm:tpof_explicit} provides a state-space realization for the two-player problem similar to Lemma~\ref{lem:centralized_explicit}, with an important difference. In Lemma~\ref{lem:centralized_explicit}, the recursions~\eqref{eq:sol_Lt} and~\eqref{eq:sol_Kt} can be solved independently by propagating time forward or backward respectively. However, the recursions~\eqref{sigma_hat}--\eqref{P_hat} are coupled in an intricate way. Both~$\hat P$ and~$\hat\Sigma$ recursions contain $\hat A$, which depends on~$\hat K$ and~$\hat L$. Furthermore, the equations for $\hat K$ contains $\hat L$ and vice-versa.

Despite being nonlinear difference equations coupled across all timesteps, the recursions~\eqref{sigma_hat}--\eqref{P_hat} can be solved efficiently. In the following theorem, we show that the equations for $\hat \Sigma, \hat P, \hat L, \hat K$ can be reduced to a \emph{linear} two-point boundary-value problem and thereby solved as efficiently as~\eqref{eq:sol_Lt}--\eqref{eq:sol_Kt}.

\begin{thm}\label{thm:dirty_formulas}
In the solution to Problem~\ref{prob:TPLQG} given by~\eqref{eq:tpof_opt_1}--\eqref{eq:tpof_opt_2}, the gains $\hat L_{0:T-1}$, $\hat K_{0:T-1}$ are of the form
\[
\hat L \eqt \bmat{ M & 0 \\ \hat L^{21} & 0 }
\qquad\text{and}\qquad
\hat K \eqt \bmat{ 0 & 0 \\ \hat K^{21} & J }
\]
where $M_{0:T-1}$ satisfies the forward recursion
\begin{equation}\label{eq:sol_Mt}
\begin{aligned}
\Gamma_0 &= \Sigma_\textup{init}^{11} \\
\Gamma_+ &\eqt A_{11}\Gamma A_{11}^\tp + M(C_{11}\Gamma A_{11}^\tp + U_{11}) + W_{11} \\
M &\eqt -(A_{11}\Gamma C_{11}^\tp + U_{11}^\tp) (C_{11}\Gamma C_{11}^\tp + V_{11})^{-1}
\end{aligned}
\end{equation}
and $J_{0:T-1}$ satisfies the backward recursion
\begin{equation}\label{eq:sol_Jt}
\begin{aligned}
F_T &= P_\textup{final}^{22} \\
F &\eqt A_{22}^\tp F_+ A_{22} + (A_{22}^\tp F_+ B_{22} + S_{22}) J + Q_{22} \\
J &\eqt -(B_{22}^\tp F_+ B_{22} + R_{22})^{-1} (B_{22}^\tp F_+ A_{22} + S_{22}^\tp)
\end{aligned}
\end{equation}
Finally, $\hat\Sigma^{21}_{0:T}$, $\hat L^{21}_{0:T}$, $\hat P_{0:T-1}^{21}$, $\hat K_{0:T-1}^{21}$ satisfy the coupled forward and backward recursions
\if\MODE3
\begin{align}\label{eq:sol_coupled_gains}
&\begin{aligned}
\hat\Sigma^{21}_0 &= \Sigma_\textup{init}^{21} \\
\hat\Sigma^{21}_+ &\eqt A_J \hat\Sigma^{21} A_M^\tp 
	+ B_{22}\hat K^{21}(\Gamma - \Sigma^{11})A_M^\tp
	 + \bigl(A_{21}\Gamma - B_{22}J\Sigma^{21}\bigr)A_M^\tp
	+ U_{12}^\tp M^\tp + W_{21}\\
\hat L^{21} &\eqt -\Bigl[ A_J\hat\Sigma^{21} C_{11}^\tp 
	+ B_{22}\hat K^{21}(\Gamma - \Sigma^{11})C_{11}^\tp
	+ (A_{21}\Gamma - B_{22}J\Sigma^{21} )C_{11}^\tp + U_{12}^\tp \Bigr] \\
	&\hspace{9cm} \times (C_{11}\Gamma C_{11}^\tp + V_{11})^{-1}
\end{aligned}\\
\label{eq:sol_coupled_gains2}
&\begin{aligned}
\hat P^{21}_T &= P_\textup{final}^{21} \\
\hat P^{21} &\eqt  A_J^\tp \hat P^{21}_+ A_M
	+ A_J^\tp (F_+ - P_+^{22})\hat L^{21} C_{11}
	+ A_J^\tp\bigl(F_+ A_{21} - P_+^{21} MC_{11}\bigr)
	+ J^\tp S_{12}^\tp + Q_{21} \\
\hat K^{21} &\eqt -(B_{22}^\tp F_+ B_{22} + R_{22})^{-1} \\
	&\hspace{1cm}\times \Bigl[ B_{22}^\tp \hat P^{21}_+ A_M 
	+ B_{22}^\tp (F_+ - P_+^{22})\hat L^{21}C_{11} + S_{12}^\tp
	+ B_{22}^\tp ( F_+ A_{21} - P_+^{21}MC_{11}) \Bigr]
\end{aligned}
\end{align}
where $A_M :\eqt A_{11} + M C_{11}$ and $A_J :\eqt A_{22}+B_{22}J$.
\else
\eqref{eq:sol_coupled_gains}--\eqref{eq:sol_coupled_gains2}
together with the definitions $A_M :\eqt A_{11} + M C_{11}$ and $A_J :\eqt A_{22}+B_{22}J$.
\begin{figure*}
\setcounter{equation}{33}
\begin{align}\label{eq:sol_coupled_gains}
&\begin{aligned}
\hat\Sigma^{21}_0 &= \Sigma_\textup{init}^{21} \\
\hat\Sigma^{21}_+ &\eqt A_J \hat\Sigma^{21} A_M^\tp 
	+ B_{22}\hat K^{21}(\Gamma - \Sigma^{11})A_M^\tp
	 + \bigl(A_{21}\Gamma - B_{22}J\Sigma^{21}\bigr)A_M^\tp
	+ U_{12}^\tp M^\tp + W_{21}\\
\hat L^{21} &\eqt -\bigl( A_J\hat\Sigma^{21} C_{11}^\tp 
	+ B_{22}\hat K^{21}(\Gamma - \Sigma^{11})C_{11}^\tp
	+ (A_{21}\Gamma - B_{22}J\Sigma^{21} )C_{11}^\tp + U_{12}^\tp \bigr) (C_{11}\Gamma C_{11}^\tp + V_{11})^{-1}
\end{aligned}\\
\label{eq:sol_coupled_gains2}
&\begin{aligned}
\hat P^{21}_T &= P_\textup{final}^{21} \\
\hat P^{21} &\eqt  A_J^\tp \hat P^{21}_+ A_M
	+ A_J^\tp (F_+ - P_+^{22})\hat L^{21} C_{11}
	+ A_J^\tp\bigl(F_+ A_{21} - P_+^{21} MC_{11}\bigr)
	+ J^\tp S_{12}^\tp + Q_{21} \\
\hat K^{21} &\eqt -(B_{22}^\tp F_+ B_{22} + R_{22})^{-1} \bigl( B_{22}^\tp \hat P^{21}_+ A_M 
	+ B_{22}^\tp (F_+ - P_+^{22})\hat L^{21}C_{11}
	+ B_{22}^\tp ( F_+ A_{21} - P_+^{21}MC_{11}) + S_{12}^\tp \bigr)
\end{aligned}
\end{align}
\hrulefill
\vspace*{-5pt}
\end{figure*}
\fi
\end{thm}
\begin{proof}
This result follows from Theorem~\ref{thm:tpof_explicit} and some straightforward algebra, so we omit the details.
The recursions~\eqref{eq:sol_Mt} and \eqref{eq:sol_Jt} are obtained by simplifying the~11 block of~\eqref{sigma_hat} and the~22 block of~\eqref{P_hat}, respectively. Finally, the recursions~\eqref{eq:sol_coupled_gains} and~\eqref{eq:sol_coupled_gains2} are obtained by simplifying the~21 blocks of~\eqref{sigma_hat} and~\eqref{P_hat} respectively.
\end{proof}

Theorem~\ref{thm:dirty_formulas} reduces the coupled recursions found in Theorem~\ref{thm:tpof_explicit} to a two-point linear boundary value problem. From a computational standpoint, computing the matrices $L_t$, $M_t$, $K_t$, $J_t$ using~\eqref{eq:sol_Lt}--\eqref{eq:sol_Kt} and~\eqref{eq:sol_Mt}--\eqref{eq:sol_Jt} requires recursing through the entire time horizon. This requires $\ord(T)$ operations.

It turns out that $\hat \Sigma^{21}_t$, $\hat P^{21}_t$, $\hat L^{21}_t$, $\hat K^{21}_t$ (and consequently $\hat L_t$ and $\hat K_t$) can also be computed in $\ord(T)$. To see why, note that~\eqref{eq:sol_coupled_gains}--\eqref{eq:sol_coupled_gains2} are of the form
\if\MODE3\else\setcounter{equation}{31}\fi
\begin{equation}\label{eq:phipsi_simpler}
\begin{aligned}
\hat\Sigma^{21}_0 &= \Sigma_\textup{init}^{21} &  \hat P^{21}_T &= P_\textup{final}^{21} \\
\hat\Sigma^{21}_+ &\eqt g_1( \hat\Sigma^{21}, \hat K^{21} ) & \hat P^{21} &\eqt g_2( \hat P^{21}_+, \hat L^{21} ) \\
\hat L^{21} &\eqt g_3( \hat\Sigma^{21}, \hat K^{21} ) & \hat K^{21} &\eqt g_4( \hat P^{21}_+, \hat L^{21} )
\end{aligned}
\end{equation}
where $g_1,\dots,g_4$ are affine functions. Eliminating $\hat L^{21}_t$ and $\hat K^{21}_t$ from~\eqref{eq:phipsi_simpler} using the last row of equations,
\begin{equation}\label{eq:phipsi_simpler2}
\begin{aligned}
\hat\Sigma^{21}_0 &= \Sigma_\textup{init}^{21} &  \hat P^{21}_T &= P_\textup{final}^{21} \\
\hat\Sigma^{21}_+ &\eqt h_1( \hat\Sigma^{21}, \hat P^{21}_+ ) & \hat P^{21} &\eqt h_2( \hat\Sigma^{21}, \hat P^{21}_+ )
\end{aligned}
\end{equation}
where $h_1$ and $h_2$ are affine functions. Now let
\[
\eta :\eqt \bmat{\vecc{\hat P^{21}} \\ \vecc{\hat\Sigma^{21}_+}}
\]
where $\vecc{X}$ is the vector obtained by stacking the columns of $X$.
Then~\eqref{eq:phipsi_simpler2} is a block-tridiagonal system of the form
\[
\bmat{ I & H_1 & & & \\[-1.5mm]
       G_1 & I & \ddots & & \\[-1.5mm]
       & \ddots & \ddots & H_{T-1}\hspace{-8pt}& \\[0.5mm]
       & & G_{T-1}\hspace{-12pt} & I\hspace{-8pt} }
\bmat{ \eta_0 \\ \eta_1 \\ \vdots \\ \eta_{T-1} }
= \bmat{ c_0 \\ c_1 \\ \vdots \\ c_{T-1} }
\]
for some constant matrices $G_{1:T-1}$ and $H_{1:T-1}$ and a constant vector $c_{0:T-1}$. Equations of this form can be solved in $\ord(T)$
using for example block tridiagonal LU factorization~\cite[\S~4.5.1]{gvl}.

Therefore, the optimal controller for the two-player problem presented in Theorem~\ref{thm:tpof_explicit} can be computed with comparable effort to its centralized counterpart in Lemma~\ref{lem:centralized_explicit}.

Note that the infinite-horizon two-player problem can be solved by making suitable assumptions on the system parameters and taking limits in~Theorem~\ref{thm:dirty_formulas}. The recursions~\eqref{eq:sol_Lt}--\eqref{eq:sol_Kt} and~\eqref{eq:sol_Mt}--\eqref{eq:sol_Jt} become algebraic Riccati equations, and the coupled recursions~\eqref{eq:phipsi_simpler} become a small set of linear equations.

\section{Concluding remarks}\label{sec:conclusion}

In this paper, we used a coordinator-based approach to derive a new structural result for a two-player partially nested LQG problem. Our results generalize those from classical LQG theory in a very intuitive way. Rather than maintaining a single estimate of the state, two different estimates must be maintained, to account for the two different sets of information available. As in the centralized case, finding the optimal two-player controller requires solving forward and backward recursions for estimation and control respectively. The key difference is that the recursions for the two-player case are coupled and must be solved together. We show that these recursions can be solved as efficiently as in the centralized case, with complexity proportional to the length of the time horizon.

An effort was made to express our results in a form that showcases the duality between estimation and control. This duality is apparent in~\eqref{sigma_hat}--\eqref{P_hat},~\eqref{eq:sol_Mt}--\eqref{eq:sol_coupled_gains2}, and in the proof of Theorem~\ref{thm:tpof_explicit}. The extent of the duality observed in the solution is perhaps unexpected. Indeed, one might expect a greater burden on the second player since it receives more measurements and must correct for the estimation errors inevitably made by the first player. However, from a different perspective, one might expect a greater burden on the first player since it has more control authority and must act to influence states of the system that the second player cannot control. The second player's lack of control authority mirrors the first player's lack of estimation ability.

\if\MODE3\newpage\fi

\bibliographystyle{abbrv}
\bibliography{pn}

\end{document}